\documentclass[11pt]{article}

\usepackage[top=1.2in,bottom=1.2in,left=1in,right=1in]{geometry}  
\usepackage{graphicx}              
\usepackage{amsmath}               
\usepackage{amsfonts}              
\usepackage{amsthm}                
\usepackage{caption} 
\usepackage{multirow}
\usepackage{lineno,hyperref}
\usepackage{authblk}
\usepackage[justification=centering]{caption}

\newtheorem{thm}{Theorem}
\newtheorem{lem}{Lemma}

\title{\textbf{On discrimination between the Lindley and xgamma distributions}}

\author[a]{Subhradev Sen\thanks{E-mail ID: subhradev.stat@gmail.com}}
\author[b]{Hazem Al-Mofleh}
\author[c]{Sudhansu S. Maiti}

\affil[a]{\small{Alliance School of Business, Alliance University, Bengaluru, India.}}
\affil[b]{\small{Department of Mathematics, Tafila Technical University, Tafila, Jordan.}}
\affil[c]{\small{Department of Statistics, Visva-Bharati University, Santiniketan, India.}}
\date{}
\begin{document}
\maketitle

\begin{abstract}
For a given data set the problem of selecting either Lindley or xgamma distribution with unknown parameter is investigated in this article. Both these distributions can be used quite effectively for analyzing skewed non-negative data and in modeling time-to-event data sets. We have used the ratio of the maximized likelihoods in choosing between the Lindley and xgamma distributions. Asymptotic distributions of the ratio of the maximized likelihoods are obtained and those are utilized to determine the minimum sample size required to discriminate between these two distributions for user specified probability of correct selection and tolerance limit.
\end{abstract}

\section{Introduction}
\label{sec:1}
It is an important problem in statistics to test whether some given observations, in view of modeling, follow one of the two probability distributions. If the two distribution possess similar structural, distributional and/or survival properties, then it is quite reasonable to construct a test procedure to determine which particular distribution need to be selected in describing the data set coming from diverse areas of application.\\ 
For last one decade, Lindley distribution (Lindley,~1958) has drawn attention of the researchers for modeling survival or reliability data sets. Its properties are explicitly studied by Ghitany et al. (2008) and is shown quite flexible in modeling time-to-event data sets. Several authors have suggested significant extensions and variations of Lindley model, see for example, Sankaran (1970), Gomez and Ojeda (2011), Nadarajah et al. (2011), Bakouch et al. (2012), Ghitany et al. (2013), Shanker et al. (2013), Merovci and Sharma (2014), Nedjar and Zeghdoudi (2016), Shibu and Irshad (2016), Asgharzadeh (2017) and references therein for more details.\\
Recently, Sen et al. (2016) introduced and studied xgamma distribution and applied it in describing survival/reliability data sets. The xgamma distribution has properties analogous to Lindley distribution and has  similar mathematical form. For more better flexibility and ease of application, few extensions of xgamma distribution are been proposed in the literature (see Sen and Chandra,~2017, Sen et al.,~2017). However, the xgamma random variables are stochastically larger than those of Lindley (see Sen et al., 2018), both the distributions are the special finite mixtures of exponential and gamma distributions and both can effectively be utilized in analyzing positively skewed data sets.\\
We address the following problem in this article. Suppose an experimenter has observed $n$ data points, say $x_1, x_2,\ldots, x_n$ and he wants to use either one parameter Lindley model or one parameter xgamma model, which one will he prefer?

The problem of testing whether some given observations follow one of the two probability distributions, is quite old in the statistical literature (see Cox,~1961, 1962; Atkinson,~1969, 1970; Chambers and Cox,~1967; Chen ,~1980 and Dyer,~1973 for more details on this). We consider in this investigation the problem of discriminating between the Lindley and xgamma distributions. We use the ratio of maximized likelihood in discriminating between the two distribution functions. We obtain the asymptotic distribution of the natural logarithm of RML following the approach of White (1982a, 1982b). It is observed that the asymptotic distribution is normal and independent of unknown parameters. The asymptotic distribution can be utilized to compute the probability of correct selection (PCS), hence we also attempted to find minimum sample size required to discriminate between the two distributions for a given value of PCS.\\
The rest of the article is organized as follows.
The method of likelihood ratio is described in section~\ref{sec:2}. Asymptotic distributions of the logarithm of RML statistics under the null hypothesis are obtained in section~\ref{sec:3}. Section~\ref{sec:4} deals with the determination of minimum sample sizes. Some Monte-Carlo simulation studies are performed in section~\ref{sec:5} and real life data set is analyzed as an illustration in section~\ref{sec:6}. Finally, section~\ref{sec:7} concludes.   

\section{Ratio of maximized likelihoods}
\label{sec:2}
Let $X_{1}, X_{2}, \ldots, X_{n}$ be independent and identically distributed (i.i.d) random variables from a Lindley or from an xgamma distribution function. We shall use the following notations throughout the article:\\
The probability density function (pdf) of a Lindley distribution with parameter $\lambda$ is denoted by
\begin{align}\label{eqn:ldpdf}
f_{LD}(x;\lambda)=\frac{\lambda^2}{(1+\lambda)}(1+x)e^{-\lambda x}; \quad x>0,\lambda>0.
\end{align}
We shall denote it by $X \sim LD(\lambda)$.\\
The pdf of xgamma distribution with parameter $\theta$ is denoted by
\begin{align}\label{eqn:xgpdf}
f_{XG}(x;\theta)=\frac{\theta^2}{(1+\theta)}\left(1+\frac{\theta}{2}x^2\right)e^{-\theta x}; \quad x>0,\theta>0.
\end{align}
We shall denote it by $X \sim XG(\theta)$.\\
Assuming the data coming from $LD(\lambda)$ or $XG(\theta)$, the likelihood functions are 
$$L_{LD}(\lambda)=\prod_{i=1}^{n}f_{LD}(x_{i};\lambda) \quad \text{and} \quad L_{XG}(\theta)=\prod_{i=1}^{n}f_{XG}(x_{i};\theta)$$ respectively.\\
The ratio of maximized likelihood (RML) is defined as
\begin{align}
L=\frac{L_{LD}(\hat{\lambda})}{L_{XG}(\hat{\theta})}
\end{align}
where $\hat{\lambda}$ and $\hat{\theta}$ are the maximum likelihood estimators of $\lambda$ and $\theta$, respectively, based on $\lbrace X_{1}, X_{2}, \ldots, X_{n}\rbrace$. Taking natural logarithm of the RML, we have the log-likelihood ration statistic as
\begin{flalign}\label{eqn:logrml}
\nonumber
T&=\ln\left[\frac{L_{LD}(\hat{\lambda})}{L_{XG}(\hat{\theta})}\right]\\
&=\ln\left(\frac{\hat{\lambda}}{\hat{\theta}}\right)^{2}+\ln\left(\frac{1+\hat{\theta}}{1+\hat{\lambda}}\right)+(\hat{\theta}-\hat{\lambda})\bar{X}+\frac{1}{n}\left[\sum_{i=1}^{n}\ln(1+X_{i})-\sum_{i=1}^{n}\ln\left(1+\frac{\hat{\theta}}{2}X_{i}^{2}\right)\right].
\end{flalign} 
Here $\bar{X}$ is the arithmetic mean of $\lbrace X_{1}, X_{2}, \ldots, X_{n}\rbrace$.\\
It is to be noted that, in case of Lindley distribution, 
\begin{align} 
\hat{\lambda}=\frac{-(\bar{X}-1)+\sqrt{(\bar{X}-1)^2+8\bar{X}}}{2 \bar{X}}
\end{align}
and in case of xgamma distribution, $\hat{\theta}$ is the solution of the equation
\begin{align}
\frac{(2+\theta)}{\theta(1+\theta)}+\frac{1}{n}\sum_{i=1}^{n}\frac{X_{i}^{2}}{2\left(1+\frac{\theta}{2}X_{i}^{2}\right)}=\bar{X}.
\end{align} 
It is clear from (\ref{eqn:logrml}) that distributions of $T$'s are independent of the parameters (see Dumonceaux et al.,~1973).\\
The following procedure can be used to discriminate between Lindley and xgamma distributions.
\begin{enumerate}
\item Choose Lindley distribution if $T>0$.
\item Choose xgamma distribution if $T<0$.
\end{enumerate}

\section{Asymptotic properties of RML}
\label{sec:3}
In this section we obtain the asymptotic distributions of RML for two different cases. Let us denote the almost sure convergence by \textit{a.s.} hereafter. The following notations are used. \\
For any Borel measurable function $h(\cdot)$, we denote,
$E_{LD}(h(Y))$ : mean of $h(Y)$ under the assumption that $Y \sim LD(\lambda)$, and 
$V_{LD}(h(Y))$ : variance of $h(Y)$ under the assumption that $Y \sim LD(\lambda)$.\\
Similarly, we define $E_{XG}(h(Y))$ and $V_{XG}(h(Y))$ as mean and variance of $h(Y)$, respectively, under the assumption that $Y \sim XG(\theta)$. Moreover, if $h(\cdot)$ and $g(\cdot)$ are two Borel measurable functions, then we define along the same line $Cov_{LD}[g(Y),h(Y)]=E_{LD}[g(Y)h(Y)]-E_{LD}[g(Y)]E_{LD}[h(Y)]$ and also $Cov_{XG}[g(Y),h(Y)]$ similarly, where $Y \sim LD(\lambda)$ and $Y \sim XG(\theta)$, respectively.

\subsection{Case 1: Data are coming from Lindley distribution}
We assume that $X_{1}, X_{2}, \ldots, X_{n}$ are $n$ data points coming from $LD(\lambda)$. For proving the main result, we take the following lemma.

\begin{lem}\label{lemma1}
Under the assumption that the data are coming from $LD(\lambda)$, as $n \rightarrow \infty$, we have
\begin{enumerate}
\item[(i)] $\hat{\lambda} \rightarrow \lambda$ \textit{a.s.}, where 
$$E_{LD}\left[\ln f_{LD}(X;\lambda)\right]=\max_{\tilde{\lambda}}E_{LD}\left[\ln f_{LD}(X;\tilde{\lambda})\right]$$
\item[(ii)] $\hat{\theta} \rightarrow \tilde{\theta}$ \textit{a.s.}, where 
$$E_{XG}\left[\ln f_{XG}(X;\tilde{\theta})\right]=\max_{\theta}E_{XG}\left[\ln f_{XG}(X;\theta)\right]$$
We denote:
$$T^{*}=\ln\left[\frac{L_{LD}(\lambda)}{L_{XG}(\tilde{\theta)}}\right]$$
\item[(iii)] $\frac{1}{\sqrt{n}}\lbrace T-E_{LD}(T)\rbrace$ is asymptotically equivalent to $\frac{1}{\sqrt{n}}\lbrace T^{*}-E_{LD}(T^{*})\rbrace$
\end{enumerate}
\end{lem}
\begin{proof}
Applying the similar argument as in theorem 1 in White (1982) we can easily get the proof and hence is omitted here.
\end{proof}
We have the following theorem.
\begin{thm}
\label{theorem1}
Under the assumption that the data are coming from $LD(\theta)$, $T$ is asymptotically normally distributed with mean $E_{LD}(T)$ and variance $V_{LD}(T)=V_{LD}(T^{*})$.
\end{thm}
\begin{proof}
By central limit theorem (CLT) and from part (ii) of lemma~\ref{lemma1}, it is clear that $\frac{1}{\sqrt{n}}\lbrace T^{*}-E_{LD}(T^{*})\rbrace$ is asymptotically normally distributed with mean $0$ and variance $V_{LD}(T^{*})$. Again by the part (iii) of lemma~\ref{lemma1}, we have $T$ is asymptotically normally distributed with mean $E_{LD}(T)$ and variance $V_{LD}(T^{*})$. Hence the proof of the theorem.
\end{proof}
Now, we shall find $\tilde{\theta}$, $E_{LD}(T)$ and $V_{LD}(T)$.
We define,
\begin{align}
\label{eqn:curltheta}
\nonumber
g(\theta)&=E_{LD}\left[\ln f_{XG}(X;\theta)\right]\\
&=2\ln \theta-\ln(1+\theta)+E_{LD}\left[\ln\left(1+\frac{\theta}{2}X^2\right)\right]-\frac{\theta(2+\lambda)}{\lambda(1+\lambda)}.
\end{align}
Therefore, $\tilde{\theta}$ can be obtained by maximizing $g(\theta)$ with respect to $\theta$ or as a solution of
\begin{align}
\label{eqn:tildetheta}
\frac{2}{\tilde{\theta}}-\frac{1}{1+\tilde{\theta}}+E_{LD}\left[\frac{\frac{2}{X^2}}{{1+\frac{\tilde{\theta}}{2}X^2}}\right]-\frac{2+\lambda}{\lambda(1+\lambda)}=0.
\end{align}
It should be noted that, $\tilde{\theta}$ depends on $\lambda$, for brevity we do not make it explicit. For further development we compute $\tilde{\theta}$ for $\lambda=1$, and we denote it by $\tilde{\theta_{1}}$.   As it is seen that it is difficult to obtain an analytic solution, we solve it numerically.\\
It is to be noted that $\lim_{n\rightarrow \infty}\frac{E_{LD}(T)}{n}$ and $\lim_{n\rightarrow \infty}\frac{V_{LD}(T)}{n}$ exist.\\
We denote, $$\lim_{n\rightarrow \infty}\frac{E_{LD}(T)}{n}=AM_{LD}(\lambda)$$ and $$\lim_{n\rightarrow \infty}\frac{V_{LD}(T)}{n}=AV_{LD}(\lambda),$$ respectively.\\
Therefore, for large $n$,
\begin{align*}
\nonumber
\frac{E_{LD}(T)}{n}\approx AM_{LD}(\lambda)
=E_{LD}\left[\ln f_{LD}(X;\lambda)-\ln f_{XG}(X;\tilde{\theta})\right].
\end{align*}
Hence,
\begin{align}\label{eqn:amld}
\nonumber
AM_{LD}(\lambda)=\ln\left(\frac{\lambda}{\tilde{\theta}}\right)^{2}+\ln\left(\frac{1+\tilde{\theta}}{1+\lambda}\right)+(\tilde{\theta}-\lambda)\frac{(2+\lambda)}{\lambda(1+\lambda)}&+E_{LD}[\ln(1+X)]\\
&-E_{LD}\left[\ln\left(1+\frac{\tilde{\theta}}{2}X^{2}\right)\right] 
\end{align} 
and we also have, for large $n$, 
\begin{align*}
\frac{V_{LD}(T)}{n}&\approx AV_{LD}(\lambda)=V_{LD}\left[\ln f_{LD}(X;\lambda)-\ln f_{XG}(X;\tilde{\theta})\right]\\
&=V_{LD}\left[\ln\left(\frac{\lambda}{\tilde{\theta}}\right)^{2}+\ln\left(\frac{1+\tilde{\theta}}{1+\lambda}\right)+(\tilde{\theta}-\lambda)X+\ln(1+X)
-\ln\left(1+\frac{\tilde{\theta}}{2}X^{2}\right)\right]
\end{align*}
Therefore,
\begin{flalign}\label{eqn:avld}
\nonumber
AV_{LD}(\lambda)&=(\tilde{\theta}-\lambda)^{2}\frac{(\lambda^2+4\lambda+2)}{\lambda^2(1+\lambda)^2}+V_{LD}[\ln(1+X)]+V_{LD}\left[\ln\left(1+\frac{\tilde{\theta}}{2}X^{2}\right)\right]\\
\nonumber
&+(\tilde{\theta}-\lambda)\left[Cov_{LD}\lbrace X,\ln(1+X)\rbrace - Cov_{LD}\left\lbrace X,\ln\left(1+\frac{\tilde{\theta}}{2}X^{2}\right)\right\rbrace\right]\\&-Cov_{LD}\left[\ln(1+X),\ln\left(1+\frac{\tilde{\theta}}{2}X^{2}\right)\right]
\end{flalign}
\subsection{Case 2: Data are coming from xgamma distribution}
In this case, we assume that $X_{1}, X_{2}, \ldots, X_{n}$ are $n$ data points coming from xgamma distribution with parameter $\theta$. We have the following lemma.
\begin{lem}\label{lemma2}
Under the assumption that the data are coming from $xg(\theta)$, as $n \rightarrow \infty$, we have
\begin{enumerate}
\item[(i)] $\hat{\theta} \rightarrow \theta$ \textit{a.s.}, where 
$$E_{XG}\left[\ln f_{XG}(X;\theta)\right]=\max_{\tilde{\theta}}E_{XG}\left[\ln f_{XG}(X;\tilde{\theta})\right].$$ 

\item[(ii)] $\hat{\lambda} \rightarrow \tilde{\lambda}$ \textit{a.s.}, where 
$$E_{LD}\left[\ln f_{LD}(X;\tilde{\lambda})\right]=\max_{\lambda}E_{LD}\left[\ln f_{LD}(X;\lambda)\right].$$
We denote
$$T_{*}=\ln\left[\frac{L_{LD}(\tilde{\lambda})}{L_{XG}(\theta)}\right]$$
\item[(iii)] $\frac{1}{\sqrt{n}}\lbrace T-E_{XG}(T)\rbrace$ is asymptotically equivalent to $\frac{1}{\sqrt{n}}\lbrace T_{*}-E_{XG}(T_{*})\rbrace$.
\end{enumerate}
\end{lem}
\begin{proof}
Proof comes applying the similar argument of White (1982) in theorem 1, hence is omitted.
\end{proof}
We have the following theorem in this case.
\begin{thm}\label{theorem2}
Under the assumption that data coming from $XG(\theta)$, $T$ is asymptotically normally distributed with mean $E_{XG}(T)$ and variance $V_{XG}(T)=V_{XG}(T_{*})$.
\end{thm}
\begin{proof}
The proof comes with similar argument as given in theorem~\ref{theorem1}.
\end{proof}
Now, we shall obtain $\tilde{\lambda}$, $E_{XG}(T)$ and $V_{XG}(T)$.\\ 
Let us define,
\begin{align*}
h(\lambda)=E_{XG}\left[\ln f_{LD}(X;\lambda)\right]
=2\ln\lambda-\ln(1+\lambda)+E_{XG}[\ln(1+X)]-\lambda\frac{(3+\theta)}{\theta(1+\theta)}
\end{align*}
Therefore, $\tilde{\lambda}$ can be obtained by maximizing $h(\lambda)$ with respect to $\lambda$ or as a solution of
\begin{align}
\label{eqn:tildelambda}
\frac{2}{\tilde{\lambda}}-\frac{1}{1+\tilde{\lambda}}-\frac{3+\theta}{\theta(1+\theta)}=0.
\end{align}
It is noted that $\tilde{\lambda}$ depends on $\theta$,  we do not make it explicit for brevity. For further development we compute $\tilde{\lambda}$ for $\theta=1$, and we denote it by $\tilde{\lambda_{1}}$.  It is difficult to obtain an analytic solution, so we solve it numerically.\\
We note that, $\lim_{n\rightarrow \infty}\frac{E_{XG}(T)}{n}$ and $\lim_{n\rightarrow \infty}\frac{V_{XG}(T)}{n}$ exist. We denote, with a similar fashion as before, 
$\lim_{n\rightarrow \infty}\frac{E_{XG}(T)}{n}=AM_{XG}(\theta)$ and $\lim_{n\rightarrow \infty}\frac{V_{XG}(T)}{n}=AV_{XG}(\theta)$, respectively.\\
We have, for large $n$,
\begin{align*}
\frac{E_{XG}(T)}{n}&\approx AM_{XG}(\theta)
=E_{XG}\left[\ln f_{LD}(X;\tilde{\lambda})-\ln f_{XG}(X;\theta)\right],
\end{align*}
so that
\begin{align}\label{eqn:amxg}
\nonumber
AM_{XG}(\theta)=\ln\left(\frac{\tilde{\lambda}}{\theta}\right)^{2}+\ln\left(\frac{1+\theta}{1+\tilde{\lambda}}\right)+(\theta-\tilde{\lambda})\frac{(3+\theta)}{\theta(1+\theta)}&+E_{XG}[\ln(1+X)]\\
&-E_{XG}\left[\ln\left(1+\frac{\theta}{2}X^{2}\right)\right]. 
\end{align} 
Again, for large $n$, we have
\begin{align*}
\frac{V_{XG}(T)}{n}&\approx AV_{XG}(\theta)=V_{XG}\left[\ln f_{LD}(X;\tilde{\lambda})-\ln f_{XG}(X;\theta)\right],
\end{align*} 
to obtain
\begin{flalign}\label{eqn:avxg}
\nonumber
AV_{XG}(\theta)&=(\theta-\tilde{\lambda})^{2}\frac{(\theta^2+8\theta+3)}{\theta^2(1+\theta)^2}+V_{XG}[\ln(1+X)]+V_{XG}\left[\ln\left(1+\frac{\theta}{2}X^{2}\right)\right]\\
\nonumber
&+(\theta-\tilde{\lambda})\left[Cov_{XG}\lbrace X,\ln(1+X)\rbrace - Cov_{XG}\left\lbrace X,\ln\left(1+\frac{\theta}{2}X^{2}\right)\right\rbrace\right]\\&-Cov_{XG} \left \lbrace \ln(1+X),\ln\left(1+\frac{\theta}{2}X^{2}\right) \right \rbrace.
\end{flalign}
Note that $\tilde{\lambda}$, $\tilde{\theta}$, $AM_{LD}(\lambda)$, $AV_{LD}(\lambda)$, $AM_{XG}(\theta)$ and $AV_{XG}(\theta)$  are numerically computed by using \texttt{R 3.5.1} programming language (R Core Team, 2018). Table~\ref{tab:tab1} and Table~\ref{tab:tab2} display the values of $AM_{LD}(\lambda)$, $AV_{LD}(\lambda)$ and $\tilde{\theta}$ for different $\lambda$, and the values of $AM_{XG}(\theta)$, $AV_{XG}(\theta)$ and $\tilde{\lambda}$ for different $\theta$, respectively.
\begin{table}[h!]
    \begin{minipage}{.5\linewidth}
      \caption{\footnotesize{The values of $AM_{LD}(\lambda)$, $AV_{LD}(\lambda)$ and $\tilde{\theta}$ for Different $\lambda$}}
      \label{tab:tab1}
      \centering
      \scalebox{0.8}{
        \begin{tabular}{lcccccc}
\hline
$\lambda$ & & $AM_{LD}(\lambda)$ & & $AV_{LD}(\lambda)$ & & $\tilde{\theta}$ \\ \hline
   0.45 & & 0.00794 & & 0.01582 & & 0.59983\\
   0.55 & & 0.00738 & & 0.01450 & & 0.72360\\
   0.65 & & 0.00699 & & 0.01363 & & 0.84547\\
   0.70 & & 0.00684 & & 0.01330 & & 0.90578\\
   0.75 & & 0.00670 & & 0.01301 & & 0.96574\\
   0.78 & & 0.00663 & & 0.01286 & & 1.00154\\
   0.89 & & 0.00639 & & 0.01237 & & 1.13186\\
   0.90 & & 0.00637 & & 0.01233 & & 1.14363\\
   1.15 & & 0.00593 & & 0.01150 & & 1.43473\\
   1.16 & & 0.00591 & & 0.01147 & & 1.44626\\
   1.37 & & 0.00561 & & 0.01090 & & 1.68648\\
   1.38 & & 0.00559 & & 0.01087 & & 1.69784\\
\hline
\end{tabular}
}
    \end{minipage}%
    \begin{minipage}{.5\linewidth}
      \centering
        \caption{\footnotesize{The values of $AM_{XG}(\theta)$, $AV_{XG}(\theta)$ and $\tilde{\lambda}$ for Different $\theta$}}
        \label{tab:tab2}
        \scalebox{0.8}{
        \begin{tabular}{ccccccc}
\hline
$\theta$ & & $AM_{XG}(\theta)$ & & $AV_{XG}(\theta)$ & & $\tilde{\lambda}$ \\ \hline
  0.85 & &-0.00718 & &0.01480 & & 0.65520\\
  0.90 & &-0.00706 & &0.01456 & & 0.69686\\
  1.05 & &-0.00674 & &0.01392 & & 0.82302\\
  1.10 & &-0.00665 & &0.01372 & & 0.86544\\
  1.25 & &-0.00639 & &0.01317 & & 0.99369\\
  1.26 & &-0.00637 & &0.01314 & & 1.00229\\
  1.40 & &-0.00616 & &0.01267 & & 1.12328\\
  1.50 & &-0.00601 & &0.01236 & & 1.21035\\
  1.65 & &-0.00580 & &0.01191 & & 1.34185\\
  1.80 & &-0.00561 & &0.01149 & & 1.47435\\
  2.00 & &-0.00536 & &0.01096 & & 1.65242\\
  2.05 & &-0.00530 & &0.01084 & & 1.69716\\
\hline
\end{tabular}
}
    \end{minipage} 
\end{table}

\section{Sample size determination}
\label{sec:4}
In this section, we shall discuss a procedure to determine the minimum sample size required to discriminate between the Lindley and xgamma distributions, for a pre-specified user specific probability of correct selection (PCS).\\
It is important to know the distance (or closeness or proximity) between the two distributions, under consideration, prior to establish a discrimination strategy between them. There are several ways to measure the distance between the two probability distributions, such as, \textit{Kolmogorov-Smirnov} (K--S) distance or \textit{Helinger} distance. It is quite natural that, for a given PCS, a very large sample is required to discriminate between the distributions that are very close.\\ 
On the other hand, a small sample may be adequate to discriminate between the two distributions that are not so close. However, it is also logical that, if the two distributions under consideration  are very close then one may not need to differentiate between them from a practical angle. So, it is expected that a practitioner will pre-specify the PCS and a tolerance limit in view of the distance between the two probability distributions. The tolerance limit indicates that an user does not want to make a distinction between two probability distributions if the distance measure between them is less than the tolerance limit.\\
We can determine the minimum sample size required to discriminate between two distributions based on PCS and the tolerance limit. We apply here the very popular K-S distance to discriminate between two distribution functions.\\ 
We have observed in section~\ref{sec:3}, that for a large $n$ RML statistics follow approximately normal distribution. We use this fact along with the help of K--S distance to determine the minimum sample size required such that PCS attains a certain protection level, say $p^*$, for a given tolerance level $D^*$.
The K--S distance between two distribution functions, say $F(x)$ and $G(x)$, is defined as
\begin{align}
\sup_{-\infty<x<\infty}|F(x)-G(x)| 
\label{eq3}
\end{align}
The method is explained below for Case 1, the procedure for Case 2 follows exactly along the same line may be noted also.
We have seen that $T$ is asymptotically normally distributed with mean $E_{LD}(T)$ and variance $V_{LD}(T)$, therefore, we have the probability of correct selection (PCS) as
\begin{align}
PCS(\lambda)=Pr\left[T>0|\lambda\right]\approx 1-\Phi \left(\frac{-E_{LD}(T)}{\sqrt{V_{LD}(T)}}\right)=1-\Phi \left(\frac{-n \times AM_{LD}(\lambda)}{\sqrt{n \times AV_{LD}(\lambda)}}\right)
\end{align} 
Here $\Phi(.)$ denotes the distribution function of the standard normal variable. To determine the sample size required to attain at least a protection level $p^*$, we equate,
$$\Phi \left(\frac{-n \times AM_{LD}(\lambda)}{\sqrt{n \times AV_{LD}(\lambda)}}\right)=1-p^*$$ 
to solve for $n$ and it provides
\begin{align}
n=\frac{z_{p^*}^2\times AV_{LD}(\lambda)}{\left[AM_{LD}(\lambda)\right]^2} \label{eq1}
\end{align}
Here $z_{p^*}$ is the $100p^*$ percentile point of a standard normal distribution. $AM_{LD}(\theta)$ and $AV_{LD}(\lambda)$ are same as defined before. For $p^{*}=0.90$ and for different values of $\lambda$, the possible values for $n$ is reported in Table~\ref{tab:tab3}.\\
Along with the similar argument, the sample size $n$ required for Case 2 is obtained as
\begin{align}
n=\frac{z_{p^*}^2\times AV_{XG}(\theta)}{\left[AM_{XG}(\theta)\right]^2} \label{eq2}
\end{align}
For $p^{*}=0.90$ and for different values of $\theta$, the possible values for $n$ is reported in Table~\ref{tab:tab4}.

From Tables~\ref{tab:tab3} and~\ref{tab:tab4} it is immediate that as $\lambda$ moves away from 0.78 and $\theta$ moves away from 1.26, for a
given PCS, the required sample size decreases as expected.\\
Suppose that one would like to choose the minimum sample size needed for a given protection level $p^*$ when the distance between two distribution functions is greater
than a given tolerance level $D^*$. 
We report K--S distance between $LD(\lambda)$ and $XG(\tilde{\theta})$ for different values of $\lambda$
is reported in Table~\ref{tab:tab3}. Here $\tilde{\theta}$ is same as defined in (\ref{eqn:tildetheta}) and it has been reported in Table~\ref{tab:tab1}. Similarly, K–-S distance between $XG(\theta)$ and $LD(\tilde{\lambda})$ for different values of $\theta$ is reported in Table~\ref{tab:tab4}. Here $\tilde{\lambda}$ is same as defined in (\ref{eqn:tildelambda}) and it has been reported in Table~\ref{tab:tab2}. \\
Now, to determine the minimum sample size required to discriminate between Lindley and xgamma distributions for given $p^*$ and $D^*$. Suppose $p^{*}=0.90$ and $D^{*}=0.03$. Here tolerance level $D^{*}= 0.03$ means that the practitioner wants to discriminate between a Lindley distribution function and a xgamma distribution function only when their K–-S distance is more than $0.03$. From Table~\ref{tab:tab3}, it is clear that for
Case 1, K–-S distance will be more than $0.03$ if $0.89 \leq \lambda \leq 1.38$. 
Similarly, from Table~\ref{tab:tab4}, it is clear that for Case 2, K–-S distance will be more than $0.03$ if $1.10 \leq \theta\leq 2.05$. Therefore, if the null distribution is Lindley, then for the tolerance level $D^{*}=0.03$, one needs $n = \max \lbrace 143, 9 \rbrace =143$ to meet the PCS when $p^{*}=0.90$.\\ 
Similarly, if the null distribution is xgamma then one needs $n = \max \lbrace137, 13\rbrace =137$ to meet the PCS when $p^{*}=0.90$ and $D^{*}= 0.03$. Therefore, for the given tolerance level $0.03$ one needs $n=\max\lbrace143,137\rbrace=143$ to meet the protection level $p^{*}=0.90$ simultaneously for both the cases.

\section{Simulation studies}
\label{sec:5}
In this section, we present some simulation results to investigate the behavior of these asymptotic
results derived in Section~\ref{sec:3} for finite sample sizes. All simulations are performed using
the \texttt{R 3.5.1} programming language (R Core Team, 2018). These \texttt{R} codes are available to the reader from the authors.\\
We consider $n = 20, 40, 60, 80, 100, 400$, and for Case 1, the null distribution is Lindley and the alternative is xgamma, we consider $\lambda=0.45,0.55,0.65,0.70,0.75,0.78,0.89,0.90,1.15,1.16,$ $1.37,1.38$, for a fixed $\lambda$ and $n$, we generate a random sample of size $n$ from $LD(\lambda)$ and we replicate this process $25,000$ times.\\
Similarly, for Case 2, the null distribution is xgamma and the alternative is Lindley, we consider $\theta=0.85,0.90,1.05,1.10,1.25,1.26,1.40,1.50,1.65,1.80,2.00,2.05$.
for a fixed $\theta$ and $n$, we generate a random sample of size $n$ from $XG(\theta)$ and we replicate this process $25,000$ times.
\begin{table}[h!]
\centering
\caption{\small{The minimum sample size in (\ref{eq1}), for $p^* = 0.90$, the null distribution is Lindley, the K–-S distance between $LD(\lambda)$ and $XG(\tilde{\theta})$ for different values of $\lambda$}}%
\label{tab:tab3} \resizebox{\textwidth}{!} {
\begin{tabular}{lcccccccccccccc}
\hline
$\lambda$ &0.45  &   0.55  &  0.65   &  0.70   &  0.75   & 0.78 & 0.89 & 0.90  &   1.15 & 1.16 & 1.37& 1.38\\
$n$       &15     &   33    &  96     &  196    &  400    & 481 & 143 & 127   &   21  & 20 & 10 & 9 \\
$K-S$     &0.02158&   0.02520&  0.02776&  0.02867&  0.02943& 0.02983& 0.03098 & 0.03106&  0.03196 & 0.03193 & 0.03118 & 0.03086\\
\hline
\end{tabular}}
\end{table}
For each case and replication, we calculate the PCS based on asymptotic results derived in Section~\ref{sec:3}, furthermore, we  calculate it as derived in Section~\ref{sec:4}. The results are reported in Tables~\ref{tab:tab5} and~\ref{tab:tab6}.\\
It is clear from Tables~\ref{tab:tab5} and~\ref{tab:tab6} that as the sample size increases the PCS
increases as expected. Also, for a small sample size, $n=20$, the asymptotic results work quite well for both the cases for all possible parameter ranges. From the simulation study it is recommended that asymptotic results can be used quite effectively even when the sample size is as small as $20$ for all possible choices of the shape parameters.
\begin{table}[h!]
\centering
\caption{\small{The minimum sample size in (\ref{eq2}), for $p^* = 0.90$, the null distribution is xgamma, the K–-S distance between $XG(\theta)$ and $LD(\tilde{\lambda})$ for different values of $\theta$}}%
\label{tab:tab4} \resizebox{\textwidth}{!} {
\begin{tabular}{lcccccccccccc}
\hline
$\theta$ &0.85  &   0.90  &   1.05  &   1.10  &    1.25    &  1.26    &  1.40   &  1.50    & 1.65  &  1.80 & 2.00 & 2.05\\
$n$      &23    &   30    &   87    &   137     &   525   &    532  &     191  &    89   &    41    &   24 & 14 & 13\\
$K-S$    &0.02674 & 0.02771 & 0.02964 & 0.03007  & 0.03099  & 0.03104  & 0.03155 & 0.03175 & 0.03183 & 0.03164 &  0.03102 & 0.03087 \\
\hline
\end{tabular}}
\end{table}

\begin{table}[h!]
\centering
\caption{\small{The PCS based on Monte-Carlo Simulations and the asymptotic results (in parenthesis) when the null distribution is Lindley, for different values of $\lambda$}}%
\label{tab:tab5}
\scalebox{0.7}{
\begin{tabular}{lcccccc}
\hline
$\lambda \downarrow n \rightarrow$ & 20        & 40        & 60        & 80       & 100       & 400   \\ \hline
\multirow{2}{*}{0.45}      & 0.62020  & 0.65264     & 0.68804     & 0.71888    & 0.73360    & 0.89796	 \\
          &$(0.61118)$  &$(0.65519)$  &$(0.68763)$  &$(0.71389)$  &$(0.73613)$  &$(0.89669)$ \\  \hline
\multirow{2}{*}{0.55}     & 0.61276   & 0.65620    & 0.68752    & 0.71496    & 0.72740    & 0.89004	 \\
          &$(0.60797)$  &$(0.65083)$  &$(0.68248)$  &$(0.70818)$  &$(0.72999)$  &$(0.88982)$ \\  \hline
\multirow{2}{*}{0.65}      & 0.61848  & 0.65288     & 0.67788	    & 0.71256     & 0.73324     & 0.88344    \\
          &$(0.60560)$  &$(0.64759)$  &$(0.67866)$  &$(0.70392)$  &$(0.72540)$  &$(0.88452)$ \\  \hline
\multirow{2}{*}{0.70}      & 0.61528  & 0.65452     & 0.68052	    & 0.70812     &0.72808      & 0.88052    \\
          &$(0.60459)$  &$(0.64621)$  &$(0.67703)$  &$(0.70211)$  & $(0.72344)$ &$(0.88223)$ \\  \hline
\multirow{2}{*}{0.75}      & 0.61896  & 0.65672     & 0.67788     & 0.70492     &0.72392      & 0.88052    \\
          &$(0.60366)$  &$(0.64494)$  &$(0.67553)$  &$(0.70044)$  & $(0.72164)$ &$(0.88009)$ \\  \hline
\multirow{2}{*}{0.78}      &0.61876   & 0.65524     & 0.67648     & 0.70312     & 0.72180     & 0.87340  \\
          &$(0.60313)$  &$(0.64422)$  &$(0.67468)$  &$(0.69949)$  &$(0.72061)$  &$(0.87886)$ \\  \hline
\multirow{2}{*}{0.89}      &0.61180  &0.65256  &0.67372  &0.70356   &0.72016   &0.87680 \\
&$(0.60135)$ &$(0.64178)$ &$(0.67179)$ &$(0.69626)$  &$(0.71712)$  &$(0.87464)$ \\  \hline

\multirow{2}{*}{0.90}      & 0.61120  & 0.65236     & 0.67688	  & 0.70064     &0.72020      & 0.87176    \\
          &$(0.60120)$  &$(0.64157)$  &$(0.67154)$  &$(0.69598)$  & $(0.71682)$ &$(0.87428)$ \\  \hline
\multirow{2}{*}{1.15}   & 0.61580    & 0.64408    & 0.67444	   & 0.69276    &0.70944     & 0.86296  \\
          &$(0.59767)$  &$(0.63674)$  &$(0.66581)$  &$(0.68958)$  &$(0.70988)$  &$(0.86565)$ \\  \hline
\multirow{2}{*}{1.16}   &0.61856 & 0.64412 & 0.67528 & 0.69528 & 0.71308 & 0.86432  \\
&$(0.59754)$  &$(0.63656)$  &$(0.66560)$  &$(0.68934)$  &$(0.70962)$  &$(0.86532)$ \\  \hline
\multirow{2}{*}{1.37}  & 0.60564    & 0.63624    & 0.67340    & 0.68640    & 0.70624    & 0.86104\\
&$(0.59489)$ &$(0.63293)$ &$(0.66128)$ &$(0.68449)$ &$(0.70436)$ &$(0.85859)$ \\  \hline
\multirow{2}{*}{1.38} & 0.60676  & 0.64320    & 0.66784    & 0.68536    & 0.70592    & 0.85880 \\
&$(0.59477)$ &$(0.63276)$ &$(0.66108)$ &$(0.68427)$ &$(0.70412)$ &$(0.85827)$\\  \hline
\end{tabular}
}
\end{table}

\begin{table}[h!]
\centering
\caption{\small{The PCS based on Monte-Carlo Simulations and the asymptotic results (in parenthesis) when the null distribution is xgamma, for different values of $\theta$}}%
\label{tab:tab6}
\scalebox{0.7}{
\begin{tabular}{lcccccc}
\hline
$\theta \downarrow n \rightarrow$   & 20        & 40          & 60         & 80    & 100   & 400 \\ \hline
\multirow{2}{*}{0.85}      &0.59160  &0.63696  &0.67272  &0.69856   &0.71836   &0.88056  \\
          &$(0.60410)$ &$(0.64554)$ &$(0.67624)$ &$(0.70122)$ &$(0.72249)$ &$(0.88110)$ \\  \hline
\multirow{2}{*}{0.90}      &0.59040  &0.63612  &0.67144  &0.69668   &0.71684   &0.87856  \\
          &$(0.60319)$ &$(0.64430)$ &$(0.67478)$ &$(0.69960)$ &$(0.72073)$ &$(0.87900)$ \\  \hline
\multirow{2}{*}{1.05}      &0.58832  &0.62940  &0.66692  &0.69312   &0.71184   &0.87252  \\
          &$(0.60088)$ &$(0.64114)$ &$(0.67102)$ &$(0.69541)$ &$(0.71620)$ &$(0.87351)$ \\  \hline
\multirow{2}{*}{1.10}    &0.58948  &0.62980  &0.66380  &0.68408  &0.71192  &0.86996 \\
 &$(0.60020)$  &$(0.64021)$  &$(0.66992)$  &$(0.69417)$   &$(0.71486)$   &$(0.87187)$
 \\  \hline
\multirow{2}{*}{1.25}   &0.58692  &0.62796  &0.66324  &0.68860   &0.70764   &0.86732  \\
          &$(0.59833)$ &$(0.63764)$ &$(0.66688)$ &$(0.69077)$ &$(0.71117)$ &$(0.86728)$ \\  \hline
\multirow{2}{*}{1.26}      &0.58712 & 0.62960 & 0.66472 & 0.68848 & 0.70440 & 0.86760\\
&$(0.59821)$  &$(0.63748)$  &$(0.66669)$  &$(0.69055)$   &$(0.71094)$   &$(0.86699)$ \\  \hline
\multirow{2}{*}{1.40}      &0.58448  &0.62688  &0.65984  &0.68536   &0.70176   &0.86180  \\
          &$(0.59662)$ &$(0.63531)$ &$(0.66410)$ &$(0.68766)$ &$(0.70780)$ &$(0.86301)$  \\  \hline
\multirow{2}{*}{1.50}      &0.58004  &0.62440  &0.65500  &0.68928   &0.70168   &0.86024  \\
          &$(0.59555)$ &$(0.63383)$ &$(0.66235)$ &$(0.68569)$ &$(0.70567)$ &$(0.86027)$  \\  \hline
\multirow{2}{*}{1.65}      &0.5822 & 0.62712 & 0.65240 & 0.67924 & 0.70000 & 0.85336\\
&$(0.59400)$  &$(0.63170)$  &$(0.65982)$  &$(0.68285)$   &$(0.70258)$   &$(0.85626)$  \\  \hline          
\multirow{2}{*}{1.80}      &0.58444  &0.62228  &0.65968  &0.67316   &0.69392   &0.84964  \\
          &$(0.59251)$ &$(0.62965)$ &$(0.65737)$ &$(0.68011)$ &$(0.69960)$ &$(0.85233)$ \\  \hline
\multirow{2}{*}{2.00}      & 0.58360 & 0.62036 & 0.64884 & 0.67656 & 0.69080 & 0.84968 \\
&$(0.59059)$ &$(0.62700)$ &$(0.65422)$ &$(0.67656)$ &$(0.69574)$ &$(0.84717)$ \\  \hline
\multirow{2}{*}{2.05}      & 0.57592 & 0.62340 & 0.64804 & 0.67084  & 0.68944 & 0.84388 \\
&$(0.59012)$ &$(0.62636)$ &$(0.65345)$ &$(0.67569)$ &$(0.69479)$ &$(0.84589)$ \\  \hline
\end{tabular}
}
\end{table}

\section{Real data illustration}
\label{sec:6}
In this section we analyze two data sets and use our method to discriminate between two distribution functions.
\subsection*{Illustration-I}
As a first illustration, data set represents the number of revolution before failure of the 23 ball bearings in the life-test is considered. It was originally reported in Lawless (1982). 
When we use Lindley model, the MLE of the parameter and the corresponding log-likelihood value are obtained as $0.0273$ and $-115.7356$, respectively. The K--S distance between the data and the fitted Lindley distribution function is 0.19283 and the corresponding p-value is 0.31739.\\
On the other hand, when we use xgamma model, the MLE of the parameter and the corresponding log-likelihood value are obtained as $0.04071$ and $-113.9634$, respectively. The K--S distance between the data and the fitted xgamma distribution function is 0.13228 and the corresponding p-value is 0.76796. We also present the observed, expected frequencies for different groups and the corresponding $\chi^2$ statistics for both the distributions to the fitted data. We observe that for this data set, both the distributions provide quite good fit to the data. The results are presented in Table~\ref{tab:chidata1}.\\
The $\chi^2$ values are $3.0419$ and $1.4667$ for Lindley and xgamma distributions with corresponding p-values 0.3852 and 0.689 respectively.\\
The logarithm of RML i.e., $T=-1.7722$ which is less than $0$. Hence, it indicates to choose the xgamma model.

\subsection*{Illustration-II}
Now, we represent another data set for second illustration. The second data set represents the waiting times (in minutes) before service of 100 bank customers (see Ghitany et al.,~2008). 
For the second data set, when we use Lindley model, the MLE of the parameter and the corresponding log-likelihood value are obtained as $0.1866$ and $-319.0374$, respectively. The K--S distance between the data and the fitted Lindley distribution function is 0.06768 and the corresponding p-value is 0.74946.\\
When we use xgamma model, the MLE of the parameter and the corresponding log-likelihood value are obtained as $0.2634$ and $-321.0203$, respectively. The K--S distance between the data and the fitted xgamma distribution function is 0.06249 and the corresponding p-value is 0.82970.\\
In this case also we present the observed, expected frequencies for different groups and the corresponding $\chi^2$ statistics for both the distributions to the fitted data and we observe that, for this data set also, both the distributions provide quite good fit. The results are presented in Table~\ref{tab:chidata2}.
In this case, the $\chi^2$ values are $0.1833$ and $1.3041$ for Lindley and xgamma distributions with corresponding p-values 0.9802 and 0.7281 respectively.\\
The logarithm of RML i.e., $T=1.9829$ which is greater than $0$. Hence, for this case the indication is to choose the Lindley model.

\begin{table}[h!]
    \begin{minipage}{.5\linewidth}
     \centering
\caption{\small{Observed and expected frequencies \\for two distributions of ball bearing data}}
\label{tab:chidata1}
\scalebox{0.8}{
\begin{tabular}{lccc}
\hline
Intervals &Observed  &\multicolumn{2}{c}{Expected frequency}\\
 &frequency &Lindley &xgamma\\
\hline
0--35 		&3		&5.93 		&4.50\\
35--55		&7		&4.46		&4.88\\
55--80		&5		&4.52		&5.45\\
80--100		&3		&2.61		&3.12\\
Above 100	&5		&5.48		&5.05\\
\hline
\end{tabular}
}
    \end{minipage}%
    \begin{minipage}{.5\linewidth}
\centering
\caption{\small{Observed and expected frequencies \\for two distributions of bank customer data}}
\label{tab:chidata2}
\scalebox{0.8}{
\begin{tabular}{lccc}
\hline
Intervals &Observed &\multicolumn{2}{c}{Expected frequency}\\
 &frequency &Lindley &xgamma\\
\hline
0--5			&31		&29.73		&26.88\\
5--10			&31		&30.46		&31.26\\
10--15			&19		&19.36		&22.03\\
15--20			&10		&10.52		&11.51\\
Above 20		&9		&9.93		&8.32\\
\hline
\end{tabular}
}
    \end{minipage} 
\end{table}
\section{Concluding remarks}
\label{sec:7}
In this article, problem of discriminating between the two families of probability distributions, viz., the Lindley and xgamma, is considered. The statistic based on the ratio of the maximized likelihoods is considered and we obtain the asymptotic distributions of the test
statistics under null hypotheses. We compare the probability of correct selection using Monte-Carlo simulations with the asymptotic results and it is observed the asymptotic results work quite well for a wide range of the parameter space, even when the sample size is very small. Therefore, the asymptotic results can be utilized to estimate the probability of correct selection. To calculate the minimum sample size required for a user specified probability of correct selection, asymptotic results are used. The concept of tolerance level is used based on the distance between the two distribution functions. For a particular $D^*$
tolerance level the minimum sample size $n$ is obtained for a given user specified protection level. Tables are provided for the protection level $0.90$, however, for the other protection level the tables can be easily generated, for example, if we need the protection level $p^* = 0.85$, then all the entries corresponding to the row of $n$, will be multiplied by $\frac{z_{0.85}}{z_{0.90}}$. Tables~\ref{tab:tab3} and \ref{tab:tab4}, therefore, can be used for any given protection level.


\begin{thebibliography}{9}

\bibitem{asgh2017} Asgharzadeh, A., Nadarajah, S., and Sharafi, F. (2017). Generalized inverse Lindley distribution with application to Danish fire insurance data, Communications in Statistics - Theory and Methods,
46, 5001-5021.

\bibitem{atkin1969} Atkinson, A. C. (1969). A test for discriminating between models. Biometrika, 56(2), 337-347.

\bibitem{atkin1970} Atkinson, A. C. (1970). A method for discriminating between models. Journal of the Royal Statistical Society. Series B (Methodological), 323-353.

\bibitem{bak2012} Bakouch, H. S., Al-Zahrani, B. M., Al-Shomrani, A. A.,  Marchi, V. A., and Louzada, F. (2012). An extended Lindley distribution. Journal of the Korean Statistical Society, 41(1), 75-85.

\bibitem{cham1967} Chambers, E. A. and Cox, D. R. (1967). Discrimination between alternative binary response models. Biometrika, 54(3-4), 573-578.

\bibitem{chen1980} Chen, W. (1980). On the tests of separate families of hypotheses with small sample size. Journal of Statistical Computation and Simulation, 11(3-4), 183-187.

\bibitem{cox1961} Cox, D. R. (1961). Tests of separate families of hypotheses. In Proceedings of the fourth Berkeley symposium on mathematical statistics and probability, Berkley, University of California Press, 105-123.

\bibitem{cox1962} Cox, D. R. (1962). Further results on tests of separate families of hypotheses. Journal of the Royal Statistical Society. Series B (Methodological), 406-424.

\bibitem{dumo1973} Dumonceaux, R., Antle, C. E., and Haas, G. (1973). Likelihood Ratio Test for discriminating between two models with unknown location and scale parameters. Technometrics, 15(1), 19-27.

\bibitem{dyer1973} Dyer, A. R. (1973). Discrimination procedures for separate families of hypotheses. Journal of the American Statistical Association, 68(344), 970-974.

\bibitem{ghit2013}  Ghitany,M. E., Al-Mutairi,D. K.,   Balakrishnan, N., and Al-Enezi, L. J. (2013). Power Lindley distribution and associated inference, Computational Statistics and Data Analysis, 64, 20-33.

\bibitem{ghit2008} Ghitany, M. E., Atieh, B., and Nadarajah, S. (2008). Lindley distribution and its application. Mathematics and computers in simulation, 78(4), 493-506.

\bibitem{gom2011} Gomez D.E. and Ojeda E.C.(2011). The discrete Lindley distribution:properties and applications, Journal of Statistical Computation and Simulation, 81(11), 1405-1416.

\bibitem{law1982}Lawless, J.F. (1982), Statistical Models and Methods for Lifetime Data, John Wiley and Sons, New York.

\bibitem{lind1958} Lindley, D. V. (1958). Fiducial distributions and Bayes' theorem. Journal of the Royal Statistical Society. Series B (Methodological), 102-107.

\bibitem{mero2014} Merovci, F. and Sharma, V. K. (2014). The beta-lindley distribution: properties and applications. Journal of Applied Mathematics, 2014, 10 pages.

\bibitem{nad2011} Nadarajah, S., Bakouch, H. S., and Tahmasbi, R. (2011). A generalized Lindley distribution. Sankhya B-Applied and Interdisciplinary Statistics, 73(2), 331-359.

\bibitem{ned2016} Nedjar, S. and Zeghdoudi, H. (2016). On gamma Lindley distribution: Properties and simulations. Journal of Computational and Applied Mathematics, 298, 167-174.

\bibitem{san1970} Sankaran M.(1970). The Discrete Poisson-Lindley Distribution, Biometrics, 26(1), 145-149.

\bibitem{sen2017a} Sen, S. and Chandra, N. (2017). The quasi xgamma distribution with application in bladder cancer data. Journal of Data Science, 15(1), 61-76.

\bibitem{sen2017b} Sen, S., Chandra, N. and Maiti, S. S. (2017). The weighted xgamma distribution: Properties and application, Journal of Reliability \& Statistical Studies, 10(1), 43-58.

\bibitem{sen2018} Sen, S., Chandra, N., and Maiti, S. S. (2018). Survival estimation in xgamma distribution under progressively type-II right censored scheme, Model Assisted Statistics and Applications, 13(2), 107-121. 

\bibitem{sen2016} Sen, S., Maiti, S. S., and Chandra, N. (2016). The xgamma Distribution: Statistical Properties and Application, Journal of Modern Applied Statistical Methods, 15(1), 774-788.

\bibitem{shankar2013} Shanker R., Sharma S., and Shanker R.(2013): A Two-Parameter Lindley Distribution for Modeling Waiting and Survival Times Data, Applied Mathematics, Vol(4), 363-368.

\bibitem{shibu2016} Shibu, D. S. and Irshad, M. R. (2016). Extended New Generalized Lindley Distribution. Statistica, 76(1), 41-56.

\bibitem{white1982a} White, H. (1982a). Maximum likelihood estimation of mis-specified models, Econometrica, 50, 1-25.

\bibitem{white1982b} White, H. (1982b). Regularity conditions for Cox's test of non-nested hypotheses, Journal of Econometrics, 19, 301-318.

\end{thebibliography}
\end{document}